\documentclass[aps,pre,preprint,groupedaddress]{revtex4-1}

\usepackage{Macros}

\renewcommand{\refeq}[1]{(\ref{#1})}

\usepackage{amsthm}
\theoremstyle{plain}
\newtheorem{theorem}{Theorem}

 \usepackage[utf8]{inputenc} 
\usepackage[T1]{fontenc}

\begin{document}


\title{Partial synchronization phenomena in  networks\\ 
of identical oscillators with non-linear coupling}


\author{Celso Freitas}
\email[]{cbnfreitas@gmail.com}

\author{Elbert Macau}
\email[]{elbert.macau@inpe.br}

\affiliation{Associate Laboratory for Computing and Applied Mathematics - 
LAC, Brazilian National Institute for Space Research - INPE, Brazil}

\author{Arkady Pikovsky}
\email[]{pikovsky@uni-potsdam.de}
\affiliation{Department of Physics and Astronomy, University of Potsdam,
Germany}
\affiliation{Department of Control Theory, Nizhni Novgorod State University,
Gagarin Av. 23, 606950,
Nizhni Novgorod, Russia}


\date{\today}

\begin{abstract}

We study a Kuramoto-like model of coupled identical phase oscillators on a network, where attractive and repulsive couplings are balanced dynamically due to nonlinearity in interaction. Under a week force, an oscillator tends to follow the phase of its neighbors, but if an oscillator is compelled to follow its peers by a sufficient large number of cohesive neighbors, then it actually starts to act in the opposite manner, i.e. in anti-phase with the majority. Analytic results yield that if the repulsion parameter is small enough in comparison with the degree of the maximum hub, then the full synchronization state is locally stable. Numerical experiments are performed to explore the model beyond this threshold, where the overall cohesion is lost. We report in detail partially synchronous dynamical regimes, like stationary phase-locking, multistability, periodic and chaotic states. Via statistical analysis of different network organizations like tree, scale-free, and random ones, and different network sizes, we found a measure allowing one to predict relative abundance of partially synchronous stationary states in comparison to time-dependent ones. 
\end{abstract}

\pacs{}
\keywords{Synchronization, complex networks, Kuramoto model}

\maketitle

\textbf{Lead paragraph: Large population of coupled oscillators synchronizes
if the coupling is attractive and desynchronizes if it repulsive. However,
the nature of coupling may depend on the strength of the  compelling field:
if the force on the oscillator becomes strong, it can switch its behavior from
a ``conformist'' to a ``contrarian''. We study such a population on a network.
Here the oscillators connected to many others become contrarians first, so that
synchrony breaks. We show that the regimes of appearing partial synchrony can be 
rather complex, with a large degree of multistability, and suggest a network 
measure which allows to predict relative abundance of static and dynamic 
regimes.}

\section{Introduction}

In a seminal work \cite{kuramoto75}, aiming to understand synchronization
phenomena, Kuramoto proposed a mathematical model of non-identical, nonlinear
phase-oscillators, mutually coupled via a common mean field. 
Studying this system, he identified a synchronization transition to an 
oscillating global mode when the coupling strength is larger than a 
critical value, which is proportional to the range of the distribution 
of the natural frequencies. Over the time, subsequent outcomes based on Kuramoto
propositions have shown that his approach can be used as a framework 
to several natural and technological systems where an ordered behavior
(synchronization) emerges from the interactions of a large number of dynamical
agents
\cite{Acebron-etal-05,Strogatz2000}.
Furthermore, works have shown that the Kuramoto model can be exploited as a
building block to develop highly efficient strategies to process information
\cite{Follmann2014,Vassilieva2011}.

Recently, generalizations of the Kuramoto model toward interconnections
between the elements more complex than the mean field one, have received
considerable attention. 
Indeed, in many real-world problems, each dynamical agent interacts with a subset of the whole ensemble \cite{Dorfler2013,Leonard2014,Sadilek2014}, which can better described using a network topology. A variety of studies have analyzed the onset of the synchronization regime in this context. For a general class of linearly coupled identical oscillators, the Master Stability Function, originally proposed by Pecora and Carrol~\cite{Pecora1998}, allows one to determine an interval of coupling strength values that yields complete synchronization, as a function of the eigenvalues of Laplacian matrix of the coupling graph. For networks of oscillators with non-identical natural frequencies, Jadbabaie et al. \cite{Jad04} were able to give similar bounds for the coupling strength of the Kuramoto model without the assumption of infinitely many phase-oscillators. 
Among related works, Ref. \cite{Franci2010} deals with a model whose natural frequency oscillators change with time, even when they are isolated. Ref. \cite{Papachristodoulou2010} explores the effects of delay in the communication between oscillators. Besides, Ref. \cite{Pecora2013} builds a bridge between graph symmetry and cluster synchronization.

Taking into consideration all of these previous results, one can roughly state
that the Kuramoto transition to synchronization always happens if the coupling between
oscillators is attractive; while, when it changes to be repulsive, this
synchronization state is absent ~\cite{Petar2014,Tsimring2005}. However,
the structure of the coupling can nontrivially depend on the level of synchrony itself.
Such a dependence, called nonlinear coupling scheme, has been explored in
recent theoretical~\cite{Filatrella-Pedersen-Wiesenfeld-07,Rosenblum-Pikovsky-07,pik2009,Baibolatov2009}
and experimental~\cite{Temirbayev_etal-12,Temirbayev_etal-13} studies dealing with
setup of global coupling. The main effect here is the partial synchrony,
which establishes at moderate coupling strengths, where the coupling is
balanced between the attractive and repulsive one.

Here, we consider the effects of the non-linear
coupling on a network:  a set of
identical oscillators, which communicate via a connected simple coupling graph.
each element is forced by a (local) mean field, which
encompasses the oscillators that are connected to it.
The coupling function is tailored so that its influence is attractive, if the
local acting field is small, or repulsive, otherwise. 
This coupling strategy implies that only nodes with a large enough number
of connections may become repulsive. Thus, the hubs play a key role for the
ensemble dynamics.
A non-linear coupling parameter in the system tunes the critical quantity of
connections and how cohesive this mean field must be in order to allow this
transition.
So, our approach can be considered as a dynamical generalization of the
inhomogeneous populations of oscillators
consisting of  \textit{conformists} and
\textit{contrarians}~\cite{Hong-Strogatz-11}.
However, the type of oscillator's depends on the
force acting on it.

Overall dynamics in the model can be qualitatively understood as follows: 
Let us assume initially
that all the  mean fields are small. Then, there are only attractive
interactions (\textit{conformists}) in the system. So, in a first moment, they
start to mutually adjust their phases.
Above a threshold, the most connected oscillators start to feel a repulsive
effect that drives them away from the synchronous state. 
In other words, if an oscillator has a sufficiently large number of neighbors
and if it suffers enough cohesive pressure from them, instead of attractiveness,
it becomes a
\textit{contrarian}, wishing to be in anti-phase with the force.
Then, due to the 
repulsiveness of some nodes, other mean fields may also become smaller. 
Finally, this tendency can shift nodes to attractiveness again. As a
consequence, 
an intermediate configuration may emerge due to the balance these conflicting
tendencies in the system.

Depending on the non-linear coupling parameter, we report a variety of
qualitative dynamic behaviors. In general, for small values of the non-linear
coupling parameter, we observed full synchronization and phase-locked states.
When this parameter is increased, multistability, periodic and chaotic dynamics 
take place. 

The paper is organized as follows.
Initially, we discuss the basic details of the model in section~\ref{sModel}.
In 
section~\ref{sFull}, the analytical result about the stability of full
synchronization is presented.
Numerical experiments in section~\ref{sOthers} illustrate different possible 
regimes that the present model can display. Finally, in
section~\ref{sCorrelation}, we perform a numerical exploration to address the
correlation of stationary phase locking states with partial synchronization 
with the network parameters, by exploring different network topologies and
sizes.

\section{Model of oscillator network with nonlinear coupling} \label{sModel}

Let us consider a system of $N$ identical phase-oscillators represented by
$(\theta_1, \ldots, \theta_N) \in \TT^N$ 
coupled through a simple and connected undirected graph $A$.
The dynamics for the $i$-th oscillator is given by the 
following ordinary differential equation
\begin{equation} \label{e1}
\dot \theta_i =
\C{1 - \eps Z_i^2}  
   \Sum_{j \in \mathcal{N}_i} 
  \sin \C{\theta_j - \theta_i },
\end{equation}
where $\mathcal{N}_i$ denotes the set of neighbors of $i$ in the coupling graph
$A$. 
Equations~\refeq{e1} are formulated in the reference frame rotating with the
common
frequency of the oscillators, so that the latter one does not appear in the
equations. The time is 
normalized by the linear coupling strength. 
The main feature of our model \refeq{e1} is the \emph{nonlinear coupling
parameter}, 
$\eps \ge 0$, which modifies the coupling at each node, with $\eps=0$ a
standard 
setup of the Kuramoto model on a network is recovered. This parameter enters as
a 
coefficient at the square of \emph{the norm of the local order parameter}: 
\begin{equation}\label{eZi}
Z_i := \abs{\sum_{j \in \mathcal{N}_i} \ee^{\ii \theta_i}}.
\end{equation}
which measures the magnitude of the force acting on oscillator with index $i$. 
On the other hand, we represent the \emph{(global) order parameter} by 
\begin{equation}\label{eR}
R \ee^{\ii \psi} = \Frac{1}{N}  \sum_{i = 1}^N \ee^{\ii \theta_i},
\end{equation}
where $R \in \C[1]{0,1}$ is its norm and $\psi \in \C[01]{0,2\pi}$ is its 
phase.

We stress that $Z_i$ is not normalized (in the sense that there is no division
by 
the number of the terms in the summation, like in $R$), as 
it measures the total action of the neighbors on the $i$-th oscillator, which is
called \emph{local mean field}.  Simple calculations
show that
$Z_i^2 \in \C[1]{0,d_i^2}$, where $d_i$ denotes the degree of the $i$-th vertex,
that is, the number of incoming or outgoing connection, since the graph is
undirected. Thus, a necessary condition for a node 
to suffer repulsive coupling, i.e. $\C{1 - \eps  Z_i}^2 < 0$, is that $\eps >
d_i^{-2}$.

The introduced order parameters $R$ and $Z_i, \ldots, Z_n$ are maximal in the case of full synchronization $\theta_1 = \ldots = \theta_N$, whiles they decrease when oscillators begin to move apart from each other.

If  $\eps  Z_{max}^2 < 1$, where $Z_{\max}^2 := \max \C[2]{Z_1^2, \ldots,
Z_N^2}$, 
then all oscillator will attract each other  so that the full synchronization is
established. Next, if
$Z_{max}^2$ becomes larger than
$(\eps)^{-1}$, the corresponding oscillator begins to be repulsive related to
its 
local mean field, and the full synchronization  breaks. As a result, $Z_{max}^2$
may decrease 
and switch again the node to be attractive. Depending on the coupling graph $A$,
on the initial conditions
$\C{\theta_1^0, \ldots,\theta_1^N}$, and on the intensity of the nonlinear 
coupling parameter $\eps$, numerical simulation reveals that model \refeq{e1} 
can exhibit different qualitative behaviors.

If the largest degree in the coupling graph 
satisfies $\eps < d_{\max}^{-2}$, 
with $d_{\max} := \max \C[2]{d_1, \ldots, d_N}$, then no node can be repulsive.
We demonstrate 
in the next section via the Lyapunov analyses, that in fact this condition 
guarantees that the full synchronized state is stable.

\section{Stability of full synchronization} \label{sFull}

The basic steps to obtain the results in this section follow the main ideas from
\cite{Jad04}. We begin presenting some preliminary concepts, including
elements of the graph theory needed, and a generalized norm of the order
parameter to define our Lyapunov function.

Let $B$ be the \emph{directed incidence matrix} of a graph $A$. Thus, $B$ is a
matrix with $N$ rows and $E$ columns, where $E$ is the number of \emph{directed
edges} of the matrix. The number of \emph{undirected edges}, \ie ignoring the
direction, equals is $E/2$. The columns of $B$ represent the edges of the graph:
if the $k$-th arrow (directed edge) of the graph goes from $i$ to $j$, then the
$k$-th column of $B$ is zero, except at positions $i$ and $j$, where  $B_{ik} =
1$ and $B_{jk} = -1$.
Regarding the dynamics of the system, an arrow from 
node $i$ to node $j$ in the graph means that node $i$ influences node $j$.
Although
the directed incidence matrix is generally defined for directed graphs, it must
be emphasized that only 
undirected graphs are considered here. We abuse terminology and 
identify a graph $A$ with its \emph{adjacency matrix}, which is an $N \times N$
matrix 
where $A_{ii} = 0$; $A_{ij} = A_{ji} = 1$, if there is and edge between nodes
$i,j$; 
and $A_{ij} = A_{ji} = 0$, otherwise. So, $E=\sum_{i,j=1}^N A_{ij}$. 
Another common characterization of a graph is  its \emph{Laplacian matrix}, 
$L := \diag \C{d_1, \ldots, d_N} - A$. One can check that
$L = \nicefrac{1}{2} B B^{\top}$. A simple illustration of these concepts is 
given at Fig. \ref{fig:grafoExemplo}.

\begin{figure}
\centering
\includegraphics[scale=1]{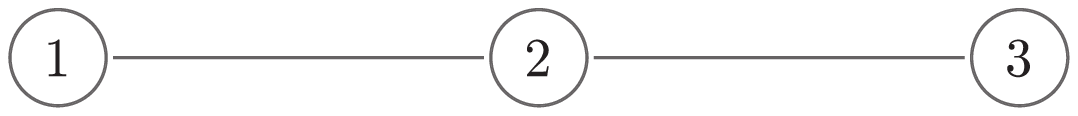}
\footnotesize
\[
{
B_1 = 
\C{ \hspace{-0.25cm} 
\
\begin{array}{rrrrrrr}
1 & -1 & 0  & 0   \\
-1 & 1 & 1  & -1  \\
0 & 0 & -1  & 1 
\end{array}
\hspace{-0.10cm} 
}
\qquad
L_1 = 
\C{ \hspace{-0.25cm} 
\
\begin{array}{rrrrrrr}
1 & -1 & 0  \\
-1 & 2 & -1 \\
0 & -1  &  1 
\end{array}
\hspace{-0.10cm} 
}
}
\]
\caption{Example of graph with $N=3$ and $E=4$, its directed incidence matrix
$B_1$ and its Laplacian matrix $L_1$.}
\label{fig:grafoExemplo}
\end{figure}
The usage of the directed incidence matrix allows us to rewrite model \refeq{e1}
in a vector form as follows:
\begin{equation} \label{e1alt}
\dot \theta
=-\Frac{1}{2} 
 \diag \C{ 1_{N} - \eps Z^2 }
B \sin \C{B^\top \theta},
\end{equation}
where $Z^2 := \C{Z_1^2, \ldots, Z_N^2}$, $1_N := \C{1, \dots, 1} \in \RR^N$ and
$\diag \C{.}$ stands for the matrix with the elements of a vector on the leading
diagonal, and $0$ elsewhere.

The square of the global order parameter can be expressed as 
\[
R^2 
= \Frac{1}{N^2} 
\C{
N
+
2\Sum_{j<k} \cos \C{ \theta_j -\theta_k }
}.
\]
However, to build our Lyapunov function,
we define a \emph{generalized norm of  order} $r$ as

\begin{equation} \label{eGenr}
r^2  :=
1 - \Frac{E - 1_E^\top 
\cos \C{B^\top \theta}  }{N^2}.
\end{equation}

Note that $R^2$ requires the sum of \emph{all} $\cos \C{\theta_j - \theta_k}$ 
with $j<k$ (for $j,k =1, \ldots, N$), but its generalization $r^2$ takes the 
sum ($1_E^\top \cos \C{B^\top \theta}$) \emph{only} through the edges of the
graph. 
In the case of full coupling graph, direct substitution yields that both global
and generalized norm of the order 
parameter have the same expression .

For any connected symmetrical coupling graph, one can check that the maximum of 
$r^2$ is the unit, and that $R^2=1$  if and only if this value is 
achieved~\footnote{On the other hand, the minimum  value of $r^2$ 
does not necessarily correspond to $R^2=0$.}.

Let
\begin{equation}\label{eU1}
U \C{\theta} = 1 - r^2
\end{equation}
be a candidate Lyapunov function. It is clear that the minimum value of 
$U \C{\theta} = 0$ corresponds to the maximum value of $r^2 = 1$, which is 
equivalent to the fully synchronized state.

In fact, algebraic manipulations reveal that
\begin{equation} \label{eU1alt}
U \C{\theta} = \frac{2}{N^2} \norm{ \sin 
\C{ \frac{ B^{\top} \theta }{2}  } 
}^2,
\end{equation}
and that the differential of $U$ is given by
\begin{equation} \label{eDU1}
\D U =
\Frac{1}{N^2} \C{ B \sin \C{B^\top \theta}}^\top.
\end{equation}

As a result, we synthesize  in the next theorem the previous intuitively stated
idea that if $\eps$ is 
small enough then full synchronization is a robust phenomenon related to small 
perturbations over initial conditions. 

\begin{theorem} \label{teo1} In Model \refeq{e1}, if
$\eps$ is smaller than a critical value $\eps_c :=
\nicefrac{1}{d_{\tiny{\mbox{max}}}^2}$,
 then the
synchronized stated ($R=1$) is Lyapunov stable.
\end{theorem}

\begin{proof} Consider the potential field $U \C{\theta}$ defined in
\refeq{eU1}. So, 
using the vector form of the model \refeq{e1alt} and the expression of the
differential 
$\D U$ from \refeq{eDU1}, we have that
$\dd*{ U \C{\theta \C{t}} }$ equals to
\[
 -\Frac{1}{2 N^2} \C{\sin \C{B^\top \theta} }^{\top} B^{\top}
 \diag \C{ 1_{N} - \eps Z^2 } 
B \sin \C{B^\top \theta},
\]
If we set $x := B \sin \C{B^\top \theta}$,  then we have that
$x^\top \diag \C{ 1_{N} - \eps Z^2 } x$ is larger or equal than $\C{1 - \eps
d^2_{\max}}
 \norm{x}^2$. Moreover, we can also define a lower bound for $\norm{x}^2$, since
$\norm{x}^2 = \sin \C{B^\top \theta}^\top B^\top B \sin \C{B^\top \theta} \ge 
 \lambda_2 \C{  B^\top B} \norm{\sin \C{B^\top \theta}}^2
 = 2
 \lambda_2 \C{L} \norm{\sin \C{B^\top \theta}}^2$; where  $\lambda_2 \C{L}$ is 
the algebraic connectivity of the graph. In the last inequality we used that
$\nicefrac{1}{2} B B^\top = L$ and that both matrices $B B^\top$ and $B^\top B$
have the same 
non-trivial eigenvalues $0 \le \lambda_2 \C{L} < \ldots < \lambda_N \C{L}$,
where $\lambda_2 
\C{L}$ is strictly larger than zero because the coupling graph $A$ is connected
\cite{Godsil01}. 
Therefore,
\[
\dd*{ U \C{\theta \C{t}} } \le -\Frac{1}{N^2}
\lambda_2 \C{L}
\C{1 - \eps d^2_{\max}}
\norm{\sin \C{B^\top \theta}}^2
\]
As a result, 
$\eps < \eps_c :=  \nicefrac{1}{d_{\tiny{\mbox{max}}}^2}$ implies that 
$\dd*{ U \C{\theta \C{t}} } \le 0$, then the fully synchronized state $R=1$ is
stable.
\end{proof}

\section{Dynamics of partially synchronous states} \label{sOthers}
In this section numerical simulations are performed to illustrate the rich
repertoire of behaviors that model \refeq{e1} may exhibit, specially beyond the 
threshold $\eps > \eps_c$, where Theorem \ref{teo1} cannot be applied. 

\subsection{Quantification of dynamical regimes} \label{sMetrics}

The numerical integration scheme applied is a forth order 
Adams-Bashforth-Moulton Method (see \cite{burden}) with discretization time step
 $h=0.01$. 
We calculate the \emph{partial synchronization metric} \textsl{s} from 
Ref.~\cite{Gardenes}, which, for every two oscillators $i,j$ in the network,  
takes values $s_{ij} \C{I} \in \C[1]{0,1}$, indicating how much the 
\emph{mean phase difference} between $\theta_i$ and $\theta_j$ varies in the 
time interval $I := \C[1]{t_1,t_2}$, with $t_1 < t_2$. This metric is defined as
\[
s_{ij} \C{I}
\doteq
\norm{
      \Frac{1}{t_2 - t_1} 
      \Int[t_1][t_2]{
      \ee^{\ii \C[0]{\theta_i \C{t} - \theta_j \C{t}}}
      }{t}
      }.
\]
One can  check that if  $\theta_i \C{t} \equiv \theta_j \C{t} + \eta$ for some 
constant $\eta$, then the exponent in the previous integral is constant and
$s_{ij} \C{I}  = 1$. 
Nevertheless, if $ \theta_i \C{t} - \theta_j \C{t} \modulo 2\pi$ assumes every
possible value 
over the unit circumference with not clear trend, then $s_{ij} \C{I}$ is close
to zero. 
Now, we average contributions of all neighbor oscillators $i,j$ under a graph
$A$ with $N$ nodes to write 
\[
	s \C{I} \doteq \Frac{1}{E} 
	\Sum_{i,j=1}^N A_{ij} s_{ij} \C{I},
\]
where $E$ is the quantity of undirected edges in the graph. 

To exclude transients and to detect the statistically stationary state, we
adopted 
the following procedure. For all experiments the time interval $\C[1]{0, 2 .
10^3}$ 
is always considered as transient time. Then, the numerical integration is
performed 
in the subsequent intervals $I_k := \C[1]{(k-1), k }10^3$, with $k \ge 3$,
until 
the first $ \tilde k = k$ such that
$\abs{ s \C{I_{\tilde k-1}} - s  \C{I_{\tilde  k}} } < 0.01$, or $\tilde k=10$
is achieved.
Only such a time interval $I_{\tilde {k}}$ is regarded as non-transient.
For the subsequent analysis, we use values of the phases $\theta \C{t}$ in the 
stationary time interval regime $I_{\tilde {k}}$ (whose beginning is shifted to
$t=0$ 
without loss of generality) at points
 $t \in \tilde I := \C[2]{i h, i \in \C[2]{0,1, \ldots, 10^5-1,10^5}}$.

\subsection{Examples of complex behaviors} \label{sExamples}

As it was claimed before, in dependence on the network structure, very
different types of the dynamics are possible. In order to give impression on it,
we present simulations of 
model \refeq{e1} with two different coupling graphs displayed as 
inserts in Fig. \ref{fig:1}. Both networks have $N=10$ nodes
and they differ only by the rewiring of a single edge. We performed simulations
for 
$10$ random initial conditions 
chosen with uniform distribution over $\C[1]{0,2\pi}$ for each experiment. For 
all these initial conditions $l = 1, \ldots, 10$, the norm of the order
parameter 
$R^l \C{t}$, according to \refeq{eR}, is computed from the time series. 
As explained in the previous section, in these calculations a transient time is
eliminated 
and that a statistically
stationary regime $\tilde I$ of $10^3$ units of time and $\#\tilde I := 10^5 +
1$ 
points is considered. Then, also for each distinct initial condition, the 
maximum, average and minimum values of the associated norm of the order
parameter are 
computed, respectively denoted by
$
R^l_{\max} : = \max_{t \in \tilde I}  R^l\C{t} 
$;
$
\C[<]{R^l} : = \C{\#\tilde I}^{-1} \sum_{t \in \tilde I} R^l\C{t} 
$; and
$
R^l_{\min} : = \min_{t \in \tilde I}  R^l\C{t} 
$. Of course, $R^l\C{t}$ converges to a constant
if and only if  $R^l_{\max} = \C[<]{R^l} = R^l_{\min}$.
Now, having different simulations for a 
fixed coupling graph, we evaluated the maximum, average and minimum value 
of the \emph{average} value of the norm of the order parameters over this
ensemble, 
respectively denoted by
$
\max \C[2]{ \C[<]{R} }: = 
\max_{l =1, \ldots, 10} \C[<]{R^l}
$;
$
\mbox{mean} \C[2]{ \C[<]{R}} : = \C{10}^{-1} \sum_{l=1, \ldots, 10} \C[<]{R^l}
$; and
$
\min \C[2]{ \C[<]{R} } : = 
\min_{l =1, \ldots, 10} \C[<]{R^l}
$. So, if the norm of the order parameter converges to the same value for
\emph{all} initial 
conditions simulated, then 
$\max \C[2]{ \C[<]{R}} =\mbox{mean} \C[2]{ \C[<]{R}} = \min \C[2]{ \C[<]{R} }$.
For the cases where the norm of the order parameter does not converge over all
initial 
conditions, it will be useful to examine the  \emph{overall maximum} and
\emph{overall minimum}
values of the norm of the order parameter, 
 respectively denoted by
$
\max \C[2]{ R_{\max} }: = 
\max_{l =1, \ldots, 10} R^l_{\max}
$;
and
$
\min \C[2]{ R_{\min} }: = 
\min_{l =1, \ldots, 10} R^l_{\min}
$. Thus, if there is no fixed phase synchronization for all the initial
conditions 
simulated, but the norm of the order parameter presents only small deviations
around a common  
value, then the gap between $\max \C[2]{ R_{\max} }$ and  $\min \C[2]{ R_{\min}
}$ is also small. 
Also notice that
$\min \C[2]{ R_{\min} } \le \min \C[2]{ \C[<]{R} } \le \mbox{mean} 
\C[2]{ \C[<]{R}} \le \max \C[2]{ \C[<]{R} } \le \max \C[2]{ R_{\max} }$, since
$R^l_{\min} 
\le \C[<]{R^l} \le R^l_{\max}$ for all initial conditions.
Finally, the maximum Lyapunov exponent $\lambda^l_{\max}$ for each initial
condition is also 
computed, according to the algorithm in \cite{alligood1997chaos}. 
The maximum Lyapunov exponent over all the chosen initial conditions 
$\lambda^l_{\max}:= \max_{l =1, \ldots, 10}  \lambda^l_{\max}$ is also analyzed.

\begin{figure}
\centering
\includegraphics[scale=1]{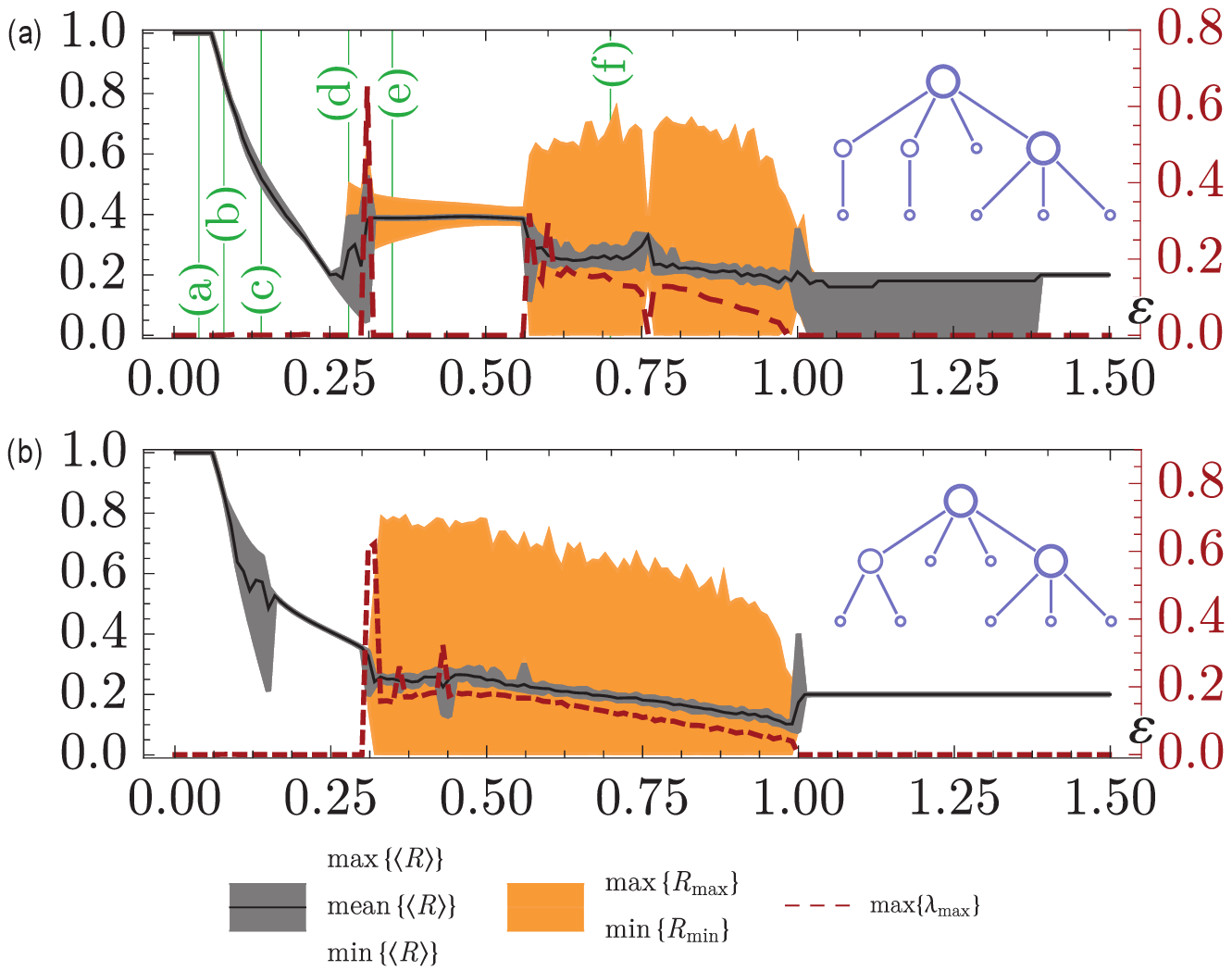}
\caption{Numerical results for Model \refeq{e1} as a function of $\eps$, for the
coupling 
graphs despited as insect, including 10 random initial conditions.  
A black line corresponds to  $\mbox{mean} \C[2]{\C[<]{R}}$, while the interval
between 
$\min \C[2]{\C[<]{R}}$ and $\max \C[2]{\C[<]{R}}$ is shown as a gray strip.
The gap between $\min \C[2]{R_{\min}}$ and  $\max \C[2]{ R_{\max}}$ is shown as
an orange strip.
Since the orange strip is by construction larger or equal than the gray one, the
first one is 
not displayed in the figure when they coincide.
Left vertical axes show values related to norm of the order parameter, while
the 
right ones represents the maximum Lyapunov exponent $\lambda_{\max}$, shown as a
red dashed line. 
Letters in green vertical lines from the upper experiment correspond to 
subfigures in Fig. \ref{fig:2}.
}
\label{fig:1}
\end{figure}

\begin{figure}
\centering
\includegraphics[scale=1]{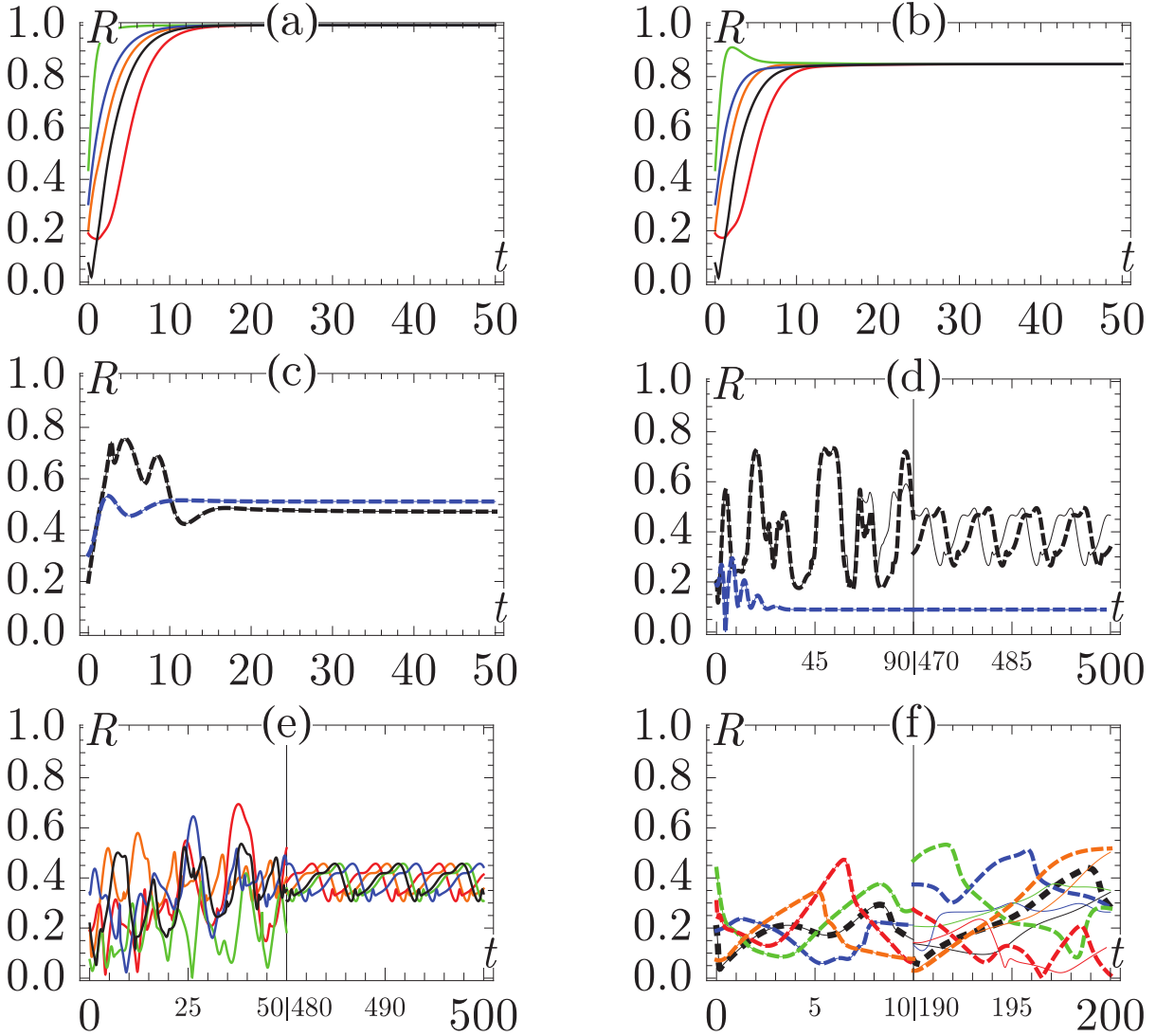}
\caption{Evolution of $R\C{t}$ for different values of $\eps$ indicated in
green 
at the upper experiment from Fig. \ref{fig:1}. Every color represents a
different 
initial condition, while pairs of solid/dashed lines with the same color
correspond 
to solutions whose initial conditions differ not more than  $10^{-4}$ at each 
coordinate. (a) $\eps=0.04$: full synchronization; 
(b) $\eps=0.08$: fixed phase synchronization;
(c,d,e) $\eps=0.15,0.28,0.35$ respect.: examples of multi stability; 
(f) $\eps=0.70$: example with $\lambda_{\max}>0$. 
}
\label{fig:2}
\end{figure}

We now describe different regimes observed in the networks, using also
 Fig. \ref{fig:2}, where we depict time series of $R\C{ \theta \C{t}}$ for some
particular choices of $\eps$, 
indicated as green letters in the upper  panel from Fig. \ref{fig:1} (this is
the case we choose for illustrating
different regimes). 
Notice that $d_{\max}=4$ in both cases, so Theorem \ref{teo1} guarantees 
that for $\eps < \eps_c = 1/4^2 = 0.0625$ the full synchronization state, $R \ra
1$,
is locally stable as illustrated in Fig. \ref{fig:2} (a) (with $\eps=0.04$). 

Panel (a) in Fig.~\ref{fig:2} illustrates full synchronization in the network
for $\eps < \eps_c$.
For $\eps$ slightly bigger than $\eps_c$, simulations suggest that a stationary
regime
of partial phase synchronization, where
$R \ra c < 1$, is locally stable as shown in Fig.~\ref{fig:2}(b) ($\eps=0.08$). 
Details of this state are clear from Fig.~\ref{fig:detailB}. There we show the
that the synchronization between
the individual oscillators is complete if measured by quantity $s_{ij}$, and all
the oscillators have
the same frequency. However,  the oscillators are split into two groups with a
constant phase shift
between them; this division originates 
in the edge which connects the two largest hubs in the network (vertexes $1,8$).

\begin{figure}
\centering
\includegraphics[scale=1]{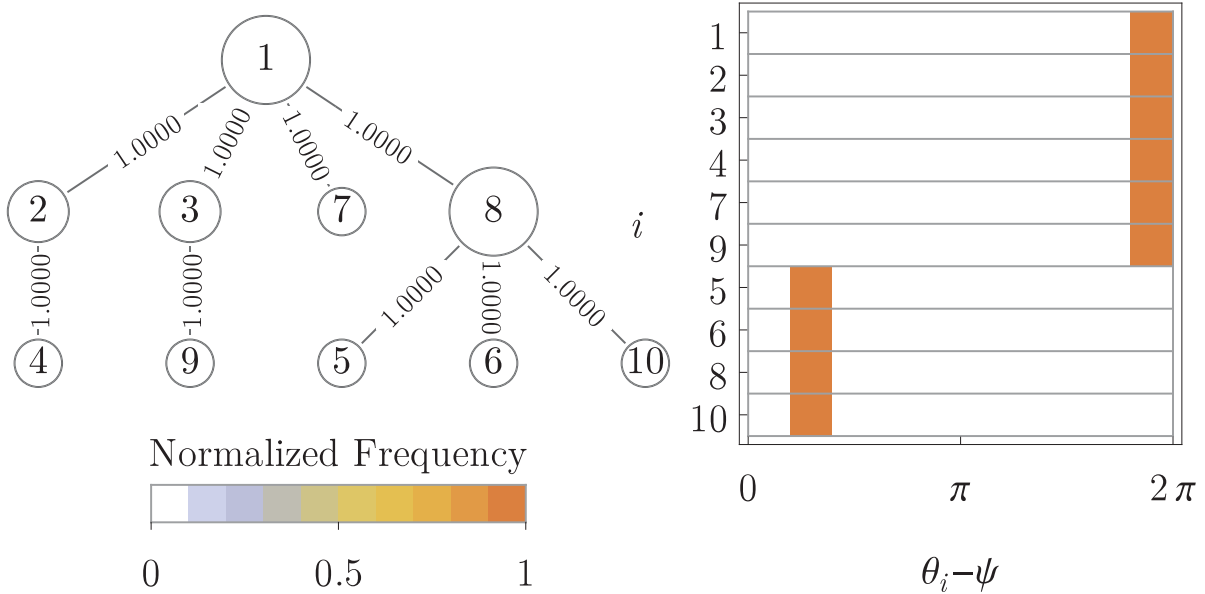}
\caption{Example of group formation.
Details of one of the trajectories from Fig.  
\ref{fig:2} (b) $\eps=0.08$. On the left side, 
the coupling graph with $s\C{i,j}$ in its edges is shown. On the right side, a
histogram of $\theta_i - \psi$ is presented with color code representing the
normalized frequency.
}
\label{fig:detailB}
\end{figure}

For larger values of $\eps$, the regimes are still static but with
multistability. 
For instance, at $\eps=0.15$  (see Fig. \ref{fig:2} (c))
two stable configurations emerge with $R \ra c$, with $c \approx 0.471$ (black)
or $c \approx 0.511$ (blue), 
depending on the initial condition. Fig.~\ref{fig:detailC}, which is analogous
to Fig.~\ref{fig:detailB},
 shows the  
existence of three subgroups, whose members may vary according to the initial
condition.

\begin{figure}
\centering
\includegraphics[scale=1]{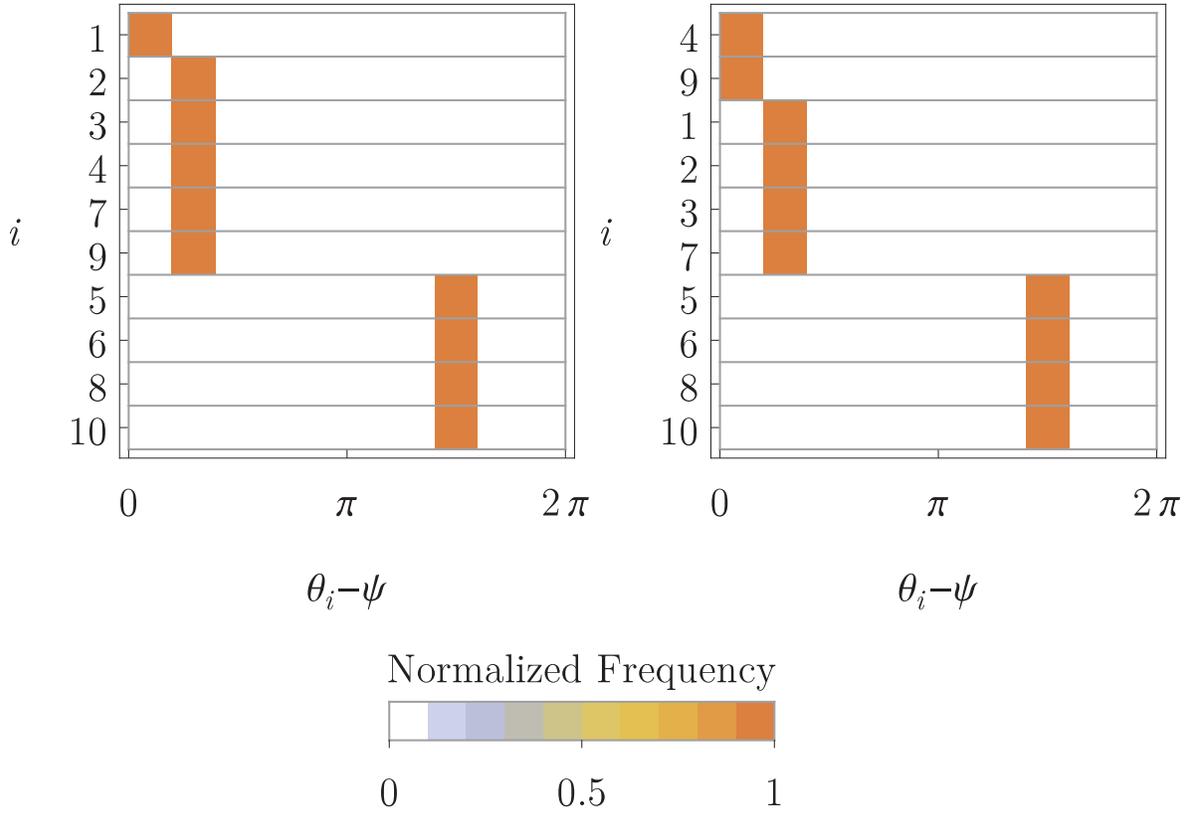}
\caption{Example multi-stability with group formation. 
Details of two trajectories from Fig.  
\ref{fig:2} (c) $\eps=0.15$,  the left picture corresponds to the solid black
line and 
the right one to the solid blue.
Histograms of $\theta_i - \psi$ are presented with color code representing the
normalized frequency.}
\label{fig:detailC}
\end{figure}

Other types of multistabilities appear for instance at $\eps=0.28$ and
$\eps=0.35$, 
as illustrate in Figs. \ref{fig:2} (d,e). For $\eps=0.28$ (panel d) some initial conditions 
do no converge to a fixed phase synchronization, but to a regime where the order
parameter 
$R$ is periodic in time.
For $\eps=0.35$ (panel c), the norm of the order parameter of all trajectories
simulated becomes periodic.
Fig. \ref{fig:detailD} provides an illustration of this regime. 

\begin{figure}
\centering
\includegraphics[scale=1]{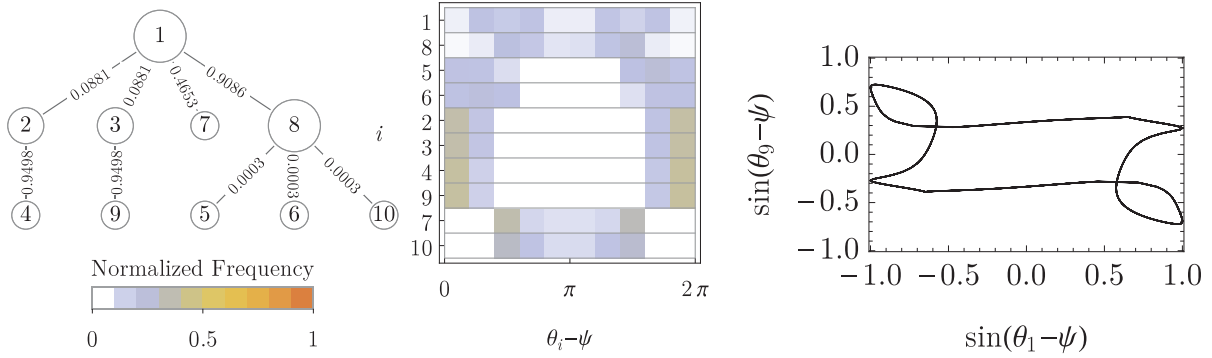}
\caption{Example of periodic norm of the order parameter.
Details of one of the trajectories from Fig.  
\ref{fig:2} (d) $\eps=0.28$. On the left side, 
the coupling graph with $s\C{i,j}$ in its edges is shown. On the right side, a
histogram of $\theta_i - \psi$ is presented with color code representing the
normalized frequency. We denote by $\psi \C{t}$ the argument of the order
parameter. The bottom picture shows that the curve $\C{  \sin \C{\theta_1 \C{t}
- \psi \C{t}}, \sin \C{\theta_9 \C{t} - \psi \C{t}}}$ is closed.
}
\label{fig:detailD}
\end{figure}

Finally, for $\eps=0.70$ (Fig. \ref{fig:2} (f)), one observes a chaotic state
with $\lambda_{\max} > 0$, the distribution of phases and frequencies is
illustrated 
in Fig.~\ref{fig:detailF}.

If $\eps \in \C[1]{1,1.5}$, we also obtained multistability, with the
coexistence of solutions converging to phase-lock and irregular order parameter
after the transient, similar to Fig. \ref{fig:2}(d).

\begin{figure}
\centering
\includegraphics[scale=1]{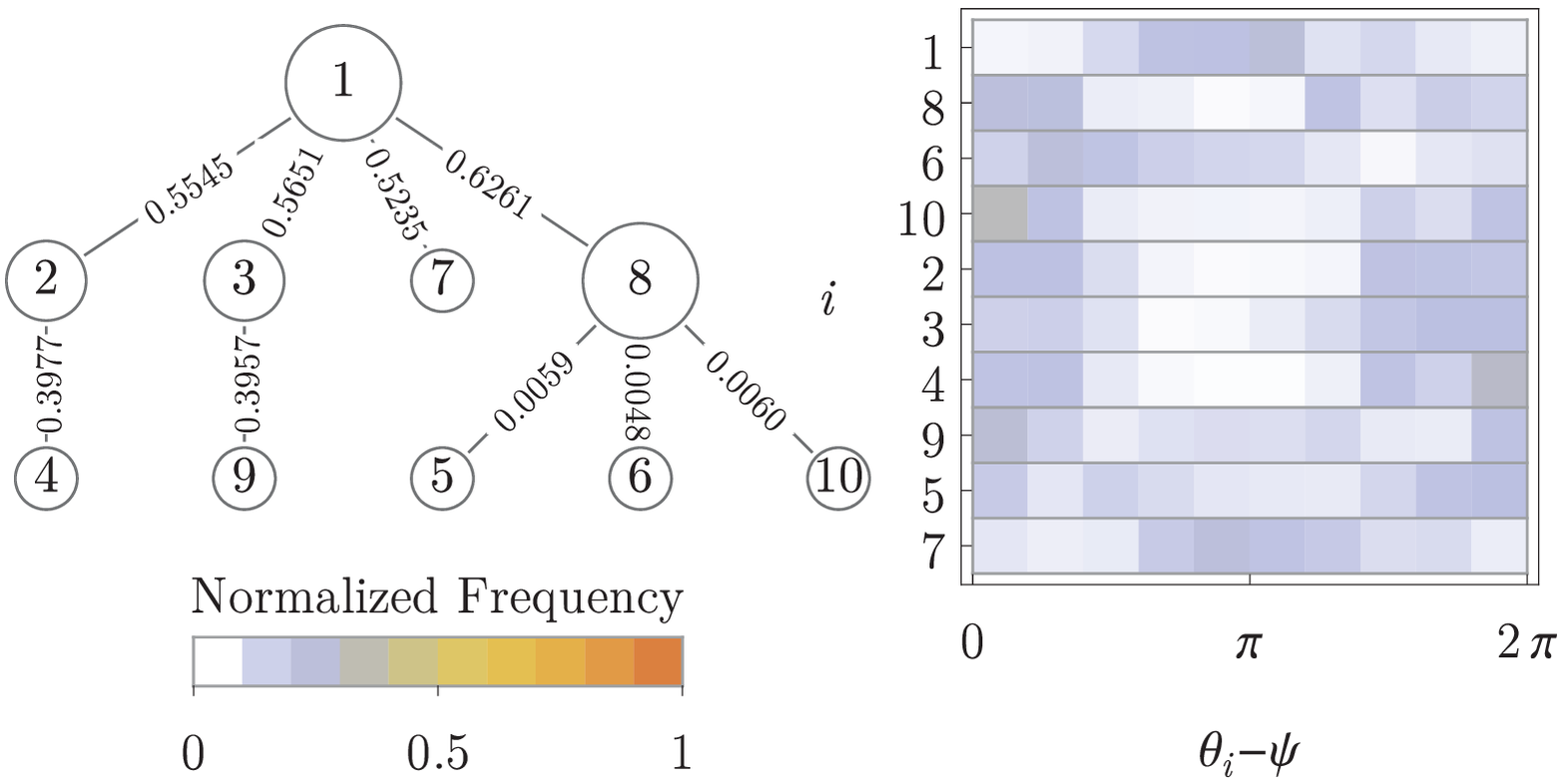}
\caption{Example of trajectory with $\lambda_{\max}>0$.
Details of one of the trajectories from Fig.  
\ref{fig:2} (f) $\eps=0.70$. On the left side, 
the coupling graph with $s\C{i,j}$ in its edges is shown. On the right side, a
histogram of $\theta_i - \psi$ is presented with color code representing the
normalized frequency.}
\label{fig:detailF}
\end{figure}

Now, we compare the results for two slightly different
networks depicted in panels (a) and (b) in Fig. \ref{fig:1}. The interval 
of values of $\eps$ with fixed phase synchronization for all initial conditions
simulated is
very similar for both networks, namely $\eps_c<\eps\lesssim 0.25$; 
also multistability of static partial synchronous regimes have been observed in both cases.

When $\eps \in \C[1]{1,1.5}$, contrary to case (a), we observed that the
solution for all initial conditions converged to the same phase-lock regime,
similar to Fig. \ref{fig:2} (b).

In conclusion of this section, in Fig. \ref{fig:OtherExperiments} we show
simulation results for 
two other networks.  Panel (a) shows a random network with $N=10$ nodes and $20$
undirected 
edges. Here predominantly static regimes are observed, only in small ranges of
coupling
constant chaos with a positive Lyapunov exponent appears. Static regimes,
however, demonstrate a 
large degree of multistability. 
In panel (b) we show a scale-free network with 
$N=50$ nodes and $100$ undirected edges. Here static states are rare, typically
irregular regimes with low
values of the order parameter are observed.

\begin{figure}
\centering
\includegraphics[scale=1]{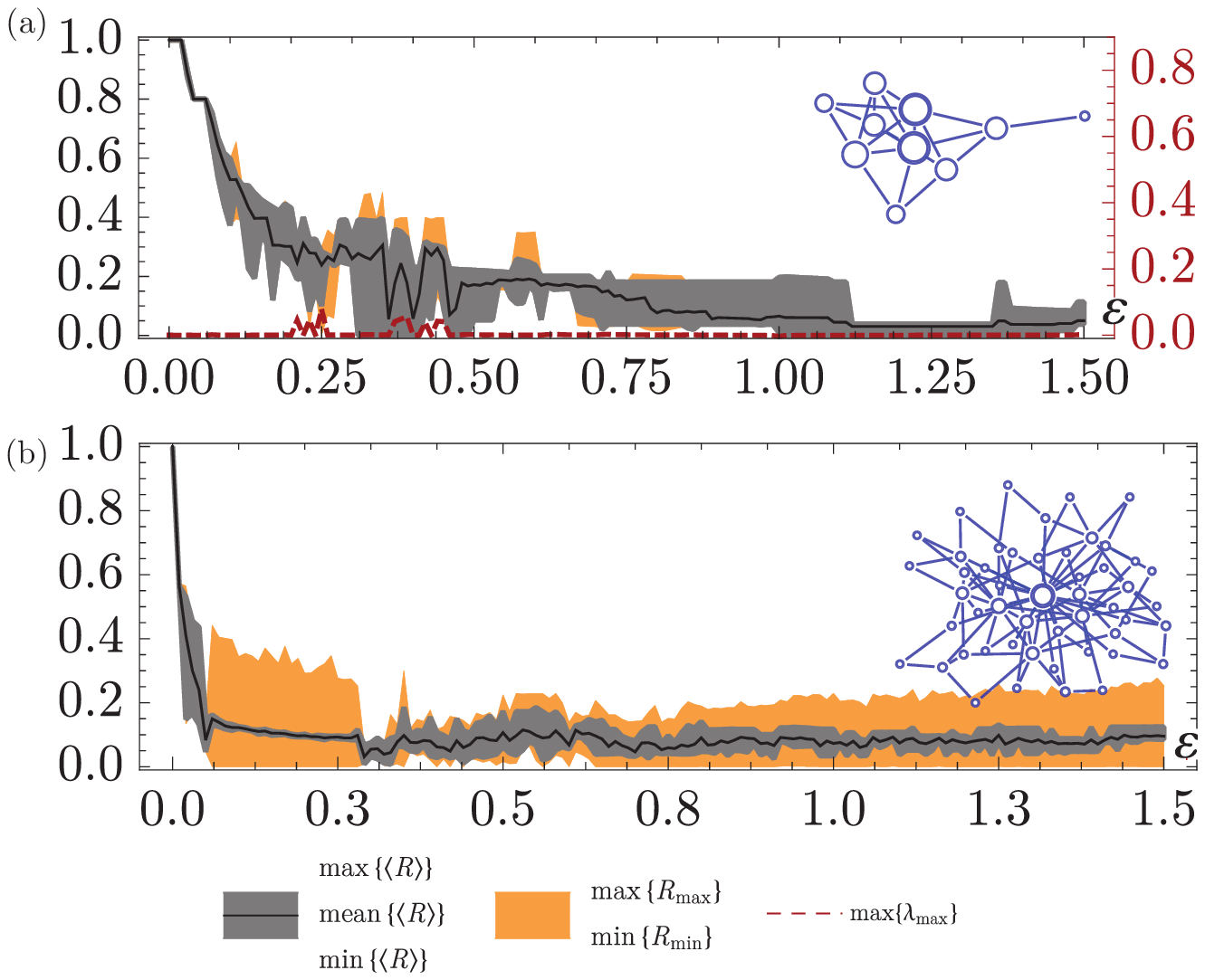}
\caption{Numerical results for Model \refeq{e1} as a function of $\eps$, for the
coupling graphs depicted as insect, including 10 random initial conditions. The
legend of the pictures is the same as in Fig. \ref{fig:1}}.
\label{fig:OtherExperiments}
\end{figure}

\subsection{Dependence of partial synchronization regimes on network structure}
\label{sCorrelation}

We have seen that partially synchronous states can be rather different even 
for similar networks. It is therefore difficult to make general predictions for
a relation between the network properties and the dynamical behaviors.
Here we attempt such a description, focusing on the property of abundance
of static regimes in comparison to time-dependent ones. For this purpose
we define the \emph{convergence index} $I_{c}$ as the ratio of values of $\eps
\in \C[1]{0, 1.5}$ 
such that $R$ converges to a constant value, 
considering  all the $10$ random initial conditions. So, both
networks Fig.~\ref{fig:1} have similar values of the index: 
$I_c \approx 0.530$ 
in case (a) while $I_c \approx 0.549$ in case (b). In contradistinction,
network shown in Fig.~\ref{fig:OtherExperiments}(a)  has very large value of the
index $I_c\approx 0.946$,
while that in Fig.~\ref{fig:OtherExperiments}(b) a rather low value $I_c\approx
0.064$.

In order to explore which features of the coupling graph are related with
$I_c$, 
we performed numerical experiments with three sets of graphs, with $N=10,50,100$
nodes. 
Each set consists in three common types of networks, each one with $10$ members,
generated as: (i)
random 
(Erd\"{o}s-R\'{e}nyi) graphs with $2N$ edges;
(ii) scale-free graphs, also with $2N$; and (iii)
tree graphs ($N$ edges). The Barab\'{a}si-Albert 
algorithm is applied for the last two types of networks (ii),(iii), with an
initial clique of $m_0$ 
nodes and with other nodes been connected to $m$ existing ones. For the
$2N$-edges scale-free 
graphs, we fixed $m_0=5$ and $m=2$; while for the tree graphs ($N$ edges
scale-free graphs), 
$m_0=m=1$. We point out that all graphs created are connected and symmetrical.
Additionally, three sets of $10$ initial conditions
$\theta_0 \in \RR^{N}$, with uniform distribution over $\C[1]{0,2\pi}$ and
$N=10,50,100$, 
have been explored. So, for each of the $90$ coupling graphs we computed its 
correspondent $I_c$ values by numerical integration of model~\refeq{e1} for
$\eps = 0,0.01,\ldots,1.49,1.50$.

In Table \ref{tab:meanIC}  we report the mean value and the standard 
deviation of $I_c$ for each topology and size of coupling graph.
From these data we see that the mean value of $I_c$ increases if we go from tree
to 
scale-free and to random graphs, respectively. However, these difference
becomes 
less noticeable for larger values of $N$. Both the mean value and the standard
deviation 
of $I_c$ decrease with larger networks. 

\begin{table}[!ht]
  \begin{center}
		\begin{tabular}{|c |c|c| c|}
		\hline Network & $N=10$  &  $N=50$  &  $N=100$  \\
		\hline Tree & 0.421 (0.260) &  0.016 (0.006) &  0.008 (0.004) \\
		\hline Scale-free & 0.857 (0.029) & 0.050 (0.015)&  0.013
(0.003) \\
		\hline Random & 0.872 (0.090) &  0.183 (0.063)& 0.077  (0.022)
\\
		\hline 
		\end{tabular} 
    \caption{Mean value of $I_c$ and its standard deviation (in brackets)
for each network type and size simulated.}
    \label{tab:meanIC}
  \end{center}
\end{table}

We have explored different networks metrics, searching for one mostly correlated
with the convergence index $I_c$. 
Let $0=\gamma_1 < \gamma_2 \le \ldots \gamma_N$ denote the Laplacian eigenvalues
of the coupling graph \cite{Godsil2001}. Recall that this graph is assumed to be
simple and connected.
We stress that these eigenvalues express fundamental characteristics of the
graph. 
For instance, $\gamma_2$ is related with graph diameter and $\gamma_N$ with its
largest degree size. 

We found that the quantity $\gamma^*$, 
defined as the ratio between the maximum eigenvalue and the average of 
the non-trivial eigenvalues of the Laplacian matrix of the graph, is rather
suitable for this purpose. Formally, it is defined as
\[
\gamma^* := 
\gamma_N
\C{ 
\frac{1}{N-1} \Sum_{k=2}^{N-1} \gamma_k
}^{-1}
.
\]

\begin{figure}
\centering
\includegraphics[scale=1]{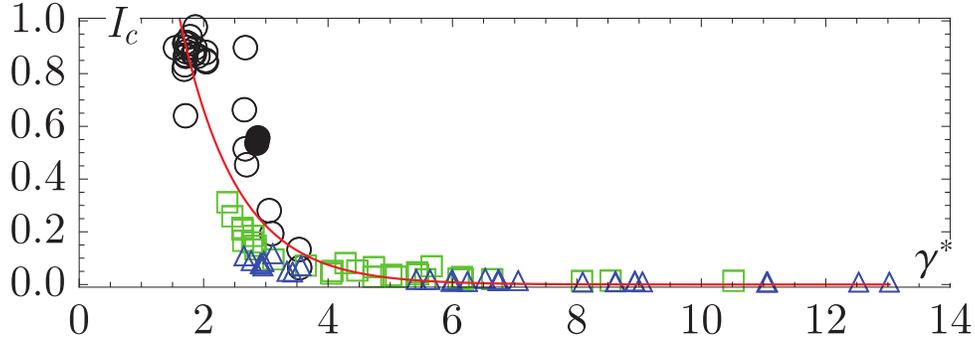}
\caption{Convergence index $I_c$ versus $\gamma^*$. Networks with $N=10,50,100$
nodes are represented as circles, squares and triangles, respectively. The two
experiments from Fig. \ref{fig:1} are shown as disks. We show in red an
exponential fit $f \C{x} = \ee^{1.7676 - 1.0894x}$ for the data.}
\label{fCorrelation}
\end{figure}

In Fig. \ref{fCorrelation} a correlation plot between $I_c$ and measure $\gamma^*$ for the correspondent graph is presented. From there, we observe a clear trend indicating that the greater the value of $\gamma^*$ is, the smaller is the value of $I_c$. Independently of the network type and size, static regimes of partial synchronization, full synchronization and phase-lock, are typical for values of $\gamma^*\lesssim 3$, like in the experiments from Fig. \ref{fig:1}. On the other hand, graphs with larger values of this measure yields more irregular dynamics, like time-dependent periodic and chaotic regimes, as the ones from Fig. \ref{fig:OtherExperiments}.

\section{Conclusion}

In this work we introduced and studied a Kuramoto-like model of identical oscillators with non-linear coupling. Our main parameter was $\eps$, which governs the coupling nonlinearity strength. It is clear that the most influence of nonlinearity in the coupling is on the hubs which experience strong forcing from many connected oscillators, while less connected nodes may still operate in a linear-coupling regime. 

We proved that if this parameter is smaller than the inverse of the square of the maximum vertex degree in the network, then the full synchronized state is stable. Via numerical experiments, we showed that our model can display a variety of other qualitative behaviors of partial synchronization, like stationary phase locking, multistability, periodic order parameter variations, and chaotic regimes. We explored the relative abundance of stationary phase locking regimes under different network topologies. Statistical analysis performed suggests that tree graphs are much less likely to exhibit stationary phase locking in comparison with scale-free or random networks. In addition, this type of behavior becomes more rare if we increase network sizes, irrespective to the network topology. Finally, we also found a good correlation between the ration between the maximum eigenvalue and the average of the non-trivial eigenvalues of the Laplacian matrix of the graph, and the proportion of the repulsion parameter values which yield stationary phase locking. The greater this measure is, the smaller tend to be presence of stationary phase locking states in the system.

As a future research, we want to investigate analytical conditions and correlations involving other graph measures related to other forms of synchronization in the model.

\begin{acknowledgments}
We would like to thank the Coordena\c{c}\~ao de Aperfei\c{c}oamento
de Pessoal de N\'ivel Superior - CAPES (Process: BEX
10571/13-2) for financial support. AP V.V. thanks the IRTG 1740/TRP 2011/50151-0, funded by the DFG/FAPESP, CNPq, and the grant/agreement 02.B.49.21.0003 of August 27, 2013 between the Russian Ministry of Education and Science and Lobachevsky State University of Nizhni Novgorod.
\end{acknowledgments}

\bibliography{References}

\end{document}